\documentclass[letterpaper, 10 pt, conference]{ieeeconf}

\IEEEoverridecommandlockouts

\usepackage{url}
\usepackage{amssymb}
\usepackage{graphicx}
\usepackage{algorithm,algorithmic}
\usepackage{tikz}
\usepackage{mathtools}
\usepackage[compress]{cite}
\usepackage{amsmath}
\usepackage{dsfont}
\usepackage{xfrac}
\usepackage{stmaryrd}
\usepackage{amsmath}
\usepackage{amssymb}
\usepackage{psfrag}

\PassOptionsToPackage{usenames,dvipsnames,svgnames}{xcolor}
\usetikzlibrary{arrows,positioning,automata,shapes,calc,shadows}

\usepackage{tcolorbox}
\definecolor{mycolor}{rgb}{0.122, 0.435, 0.698}
\makeatletter
\newcommand{\mybox}[1]{%
  \setbox0=\hbox{#1}%
  \setlength{\@tempdima}{\dimexpr\wd0+13pt}%
  \begin{tcolorbox}[colframe=black,boxrule=0.5pt,arc=4pt,left=6pt,right=6pt,top=6pt,bottom=6pt,boxsep=0pt,width=\@tempdima]
    #1
  \end{tcolorbox}
}
\makeatother

\newtheorem{definitionx}{Definition}

\newtheorem{theoremx}{Theorem}

\newtheorem{examplex}{Example}

\newcommand{\argmax}{\mathop{\text{arg\,max}}}

\definecolor{matlab_blue}{rgb}{0,0.4470,0.7410}
\definecolor{matlab_red}{rgb}{0.8500,0.3250,0.0980}

\newenvironment{definition}
  {\definitionx}
  {\hfill$\vartriangleleft$\enddefinitionx}
  
\newenvironment{example}
  {\examplex}
  {\hfill$\vartriangleleft$\endexamplex}

\newenvironment{theorem}
  {\theoremx}
  {\hfill$\vartriangleleft$\endtheoremx}

\newcommand{\range}[1]{\llbracket #1 \rrbracket}

\title{Non-Stochastic Hypothesis Testing \\ with Application to Privacy Against Hypothesis-Testing Adversary}

\author{
Farhad Farokhi\thanks{F. Farokhi is with the Data61 at the Commonwealth Scientific and Industrial Research Organisation (CSIRO) and the Department of Electrical and Electronic Engineering at the University of Melbourne, Australia}
\thanks{e-mails: farhad.farokhi@\{data61.csiro.au, unimelb.edu.au\}}
}

\begin{document}

\maketitle

\begin{abstract} In this paper, we consider privacy against hypothesis testing adversaries within a non-stochastic framework. We develop a theory of non-stochastic hypothesis testing by borrowing the notion of uncertain variables from 
non-stochastic information theory. We define tests as binary-valued mappings on uncertain variables and prove a fundamental bound on the best performance of  tests in non-stochastic hypothesis testing. We use this bound to develop a measure of privacy. We then construct reporting policies with prescribed privacy and utility guarantees. The utility of a reporting policy is measured by the distance between the reported and original values. We illustrate the effects of using such privacy-preserving reporting polices on a publicly-available practical dataset of  preferences and demographics of young individuals, aged between 15-30, with Slovakian nationality.
\end{abstract}

\section{Introduction}
For decades, stochastic policies have been used for privacy protection~\cite{warner1965randomized}. More recently, stochastic policies with provable privacy guarantees have been developed within differential privacy~\cite{dwork2014algorithmic, duchi2013local,kairouz2014extremal,machanavajjhala2008privacy} and information-theoretic privacy~\cite{farokhi2016privacy,wainwright2012privacy, liang2009information,lai2011privacy,li2015privacy, farokhi2018fisher,sankar2013utility}. Differential privacy uses randomization to ensure that the statistics of the reported outputs do not change noticeably by variations in an individual entry of the dataset. This can be ensured by the use of additive Laplace or Gaussian noise
with a scale proportional to the sensitivity of the reports on a private dataset with respect to the individual entries of the dataset.  Information-theoretic privacy, dating back to the secrecy problem~\cite{6772207}, emphasizes on masking or equivocating of information from the intended primary receiver or a secondary receiver with as much information as the primary receiver (e.g., an eavesdropper)~\cite{sankar2013utility,courtade2012information,yamamoto1983source, yamamoto1988rate} while providing guarantees on utility by bounding distortion, i.e., the distance between obfuscated and original reports.   

Although the above-mentioned stochastic policies provide provable privacy guarantees, many organizations still use deterministic heuristic-based privacy-preserving methods, such as $k$-anonymity~\cite{samarati2001protecting,sweeney2002k} and $\ell$-diversity~\cite{1617392}. For instance, anonymization is frequently used by governments\footnote{See \url{https://data.gov.au} for an example of government initiative.} or companies alike for releasing private data\footnote{See \url{https://www.kaggle.com} for examples of data from private companies and individuals.} to the broader public for analysis even though it is proved to be insufficient for privacy preservation~\cite{narayanan2008robust,su2017anonymizing, de2013unique}. Other policies, such as $k$-anonymity, are also vulnerable to attacks~\cite{1617392}. 

The popularity of non-stochastic/deterministic privacy-preserving policies is perhaps caused by factors, such as undesirable properties of differentially-private additive noise especially the Laplace noise~\cite{bambauer2013fool,farokhi2016optimal}, simplicity of implementing deterministic policies (in the sense of not requiring expertise in probability theory)~\cite{garfinkel2018issues}, and generation of unreasonable/unrealistic outputs by the use of randomness~\cite{bild2018safepub, bhaskar2011noiseless,nabar2006towards}. The guarantees of information-theoretic privacy policies are also presented in the form of averages, i.e., they only bound the average amount of leaked information. Information-theoretic policies also require the knowledge of probability distributions of datasets, which might not be available at the time of design or might change across time. This motivates the need for better understanding of non-stochastic privacy policies. 

The lack of systematic methods for developing or assessing deterministic privacy-preserving policies is due to the lack of privacy measures for deterministic policies on deterministic datasets. This makes proving  privacy guarantees for deterministic policies, in any sense, even if weak or limited in scope or practice, impossible. Recently, the theory of non-stochastic information theory~\cite{hartley1928transmission, kolmogorov1959varepsilon,renyi1961measures, nair2013nonstochastic,jagerman1969varepsilon,nair2012nonstochastic, duan2015transfer,wiese2016uncertain} was used to develop a deterministic measure of privacy~\cite{8662687}. The measure was successfully utilized to show that binning, a popular deterministic policy for privacy preservation, provides some guarantees, and to prove that $k$-anonymity is in fact \textit{not} privacy preserving, without resorting to extensive simulations and numerical studies (which are only sufficient and not necessary in analysis of general policies). The privacy measure in~\cite{8662687} is perfect for providing protections against generic adversaries; however, in some instances more might be known about the privacy-intrusive adversaries, hence the privacy measure can be further refined.

A category of adversaries studied in privacy literature is the category of hypothesis testing adversaries~\cite{barber2014privacy,li2015privacy, kairouz2014extremal}. The adversary, here, is interested in examining the validity of a hypothesis, e.g., if a house is occupied or if an individual has a certain disease based on some reports. The privacy risk, in this case, is often 
measured by the minimum error probability of the adversary. In this paper, we expand this analysis to a non-stochastic framework. To do so, we develop a theory of non-stochastic hypothesis testing by borrowing the concept of uncertain variables from 
non-stochastic information theory. Uncertain variables only consider support
sets and do not assign distributions/measures to variables. In non-stochastic hypothesis theory, we define tests as binary-valued functions on uncertain variables. We prove a fundamental bound for the best  performance of tests. This bound is used to develop a measure of privacy. We then construct reporting functions with given privacy and utility guarantees. We illustrate the outcomes of using such privacy-preserving polices on a practical dataset.

The rest of the paper organized as follows. We present the non-stochastic hypothesis testing framework in Section~\ref{sec:nonstochasticht}. In Section~\ref{sec:privacy}, we investigate privacy against hypothesis-testing adversaries. Finally, we present numerical results in Section~\ref{sec:numerical} and conclude the paper in Section~\ref{sec:conclusions}.

\section{Non-Stochastic Hypothesis Testing}
\label{sec:nonstochasticht}
In this section, we develop a framework for non-stochastic hypothesis testing starting by introducing the notion of uncertain variables.
\subsection{Uncertain Variables}
Consider $\Omega$ whose elements $\omega\in \Omega$ are samples. These elements the
source of uncertainty. An uncertain variable is a mapping defined over $\Omega$,
such as $X:\Omega\rightarrow\mathbb{X}$ with $X(\omega)$ denoting a realization of
the uncertain variable. When the dependence of the uncertain variable, u.v., to the sample
is evident from the context, $X(\omega)$ is replaced by $X$. In this paper, we
restrict ourselves to real-valued uncertain variables, e.g., $\mathbb{X}\subseteq\mathbb{R}^{n_x}$ for some integer $n_x\geq 0$. \textit{Marginal range} of any
uncertain variable $X$ is $\range{X}:=\{X(\omega):\omega\in\Omega \}\subseteq
\mathbb{X}$,  \textit{joint range} of two uncertain variables $X:\Omega\rightarrow\mathbb{X}$
and $Y:\Omega\rightarrow\mathbb{Y}$ is $\range{X,Y}:=\{(X(\omega), Y(\omega)):
\omega\in\Omega \}\subseteq \mathbb{X}\times \mathbb{Y}$, and  \textit{conditional range} of
$X$, conditioned on realization of uncertain variable $Y(\omega)=y$, is
$\range{X|y}:=\{X(\omega):\exists \omega\in\Omega  \mbox{ such that } Y(\omega)=y\}
\subseteq \range{X}.$ Uncertain variables $(X_i)_{i=1}^n$ are \textit{unrelated}
if $\range{X_1,\dots,X_n}=\range{X_1}\times \cdots \times \range{X_n}.$ Further,
they are conditionally unrelated, conditioned on uncertain variable $Y$, if $\range{X_1,\dots,X_n|y}
=\range{X_1|y}\times \cdots \times \range{X_n|y}$ for all $y\in\range{Y}$.
For two uncertain variables, $X_1$ and $X_2$ are unrelated if
 $\range{X_1|x_2}=\range{X_1}, \forall x_2\in
\range{X_2}$. 

An uncertain variable $X$ for which $\range{X}$ is  uncountably infinite is a \textit{continuous} uncertain variable, similar to a continuous random variable. An uncertain variable $X$ for which $\range{X}$ is finite is a \textit{discrete} uncertain variable. Non-stochastic entropy of a continuous uncertain variable $X$ can be defined as
\begin{align} \label{eqn:firstentropys}
h_0(X):=\log(\mu(\range{X}))\in \mathbb{R}\cup\{\pm \infty\},
\end{align}
where $\mu$ is the Lebesgue measure. The logarithm can be taken in any basis; the logarithm is in the natural basis in this paper in line with the literature on differential entropy of continuous random variables. The non-stochastic entropy in~\eqref{eqn:firstentropys} is sometimes referred to as R\'{e}nyi differential $0$-entropy~\cite{nair2013nonstochastic}. Non-stochastic entropy of a discrete uncertain variable $X$ can be defined as
\begin{align} \label{eqn:firstentropys_discrete}
H_0(X):=\log(|\range{X}|)\in \mathbb{R},
\end{align}
where $|\cdot|$ is the cardinality of a set. In this paper, for discrete uncertain variables, in line with the literature on entropy of discrete random variables, the logarithm is in the basis of two.

\begin{figure}[t]
\centering
\begin{tikzpicture}
\node[] at (-3.0,-0.3) { \small
\begin{minipage}{2cm}
\centering
hidden \\ uncertainity \\ $\omega$
\end{minipage}
};
\node[] (X) at (0,0) {$X$};
\node[] (Y) at (2,0) {$Y$};
\node[] (H) at (2,-1) {$H$};
\draw[->] (X) -- (Y);
\draw[->] (X) -- (H);
\draw[->] (-2.0,+0.0) -- (X);
\end{tikzpicture}
\caption{\label{fig:uncertainvariable} Relationship between uncertain variables in non-stochastic hypothesis testing based on uncertain measurements.}
\end{figure}
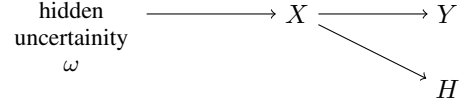

\subsection{Hypothesis Testing Based on Uncertain Measurements}
Consider three uncertain variables in Figure~\ref{fig:uncertainvariable}. Uncertain variable $X$ denotes an original uncertain variable. We have access to an \textit{uncertain measurement} of this variable denoted by $Y$. This is captured by that $Y=g_Y(X)$ for a mapping $g_Y:\range{X}\rightarrow \range{Y}$. Using the definition of unrelated variables, the relationship can be expressed by that, for any uncertain variable $Z$, $Y$ and $Z$ are conditionally unrelated, conditioned on uncertain variable $X$. Recalling that uncertain variables are mappings from the sample space, it must be that $Y=g_Y\circ X$. Similarly, we may define the \textit{hypothesis} as an uncertain variable $H$ with binary range $\range{H}=\{p_0,p_1\}$, where $p_0$ denotes the \textit{null hypothesis} and $p_1$ denotes the \textit{alternative hypothesis}. We assume that there exists a mapping $g_H:\range{X}\rightarrow\range{H}$ such that  $H=g_H\circ X$; the hypothesis is constructed based on the uncertain variable $X$ as $H=g_H(X)$.

A \textit{test} is a function $T:\range{Y}\rightarrow\range{H}=\{p_0,p_1\}$. If $T(Y)=p_1$, the test rejects the null hypothesis in favour of the alternative hypothesis; however, if $T(Y)=p_0$, the test accepts the null hypothesis. The set of all tests is given by $\range{H}^{\range{Y}}$, which captures the set of all functions from $\range{Y}$ to $\range{H}$.

Consider $y\in\range{Y}$ such that $T(y)=p_0$; the null hypothesis is accepted. The realization of output $Y(\omega)=y$ may correspond to many realizations of uncertain variable $X$, i.e., all the elements of the set $\range{X|y}$. We say that $T(y)=p_0$ is correct, or the test is correct for the output realization $Y(\omega)=y$, if $g_H(x)=p_0$ for all $x\in\range{X|y}$, i.e., all realizations of uncertain variable $X$ compatible with $y$ that are also compatible with the null hypothesis. The same also holds for the alternative hypothesis. In following definition, we use the notation that $\range{H|\range{X|y}}:=\{h\in\range{H|x}:x\in \range{X|y}\}=\cup_{x\in \range{X|y}}\range{H|x}$. 

\begin{definition}[Correctness] A test $T\in \range{H}^{\range{Y}}$ is correct at $y\in\range{Y}$ if $\range{H|\range{X|y}}=\{T(y)\}$.  The set of all outputs for which the test is correct is  $\aleph(T):=\{y\in\range{Y}:\range{H|\range{X|y}}=\{T(y)\}\}$.
\end{definition}

Based on this definition of correctness, we can define a performance measure for a test:
\begin{align}
\mathcal{P}(T):=
\begin{cases}
\log(\mu(\aleph)), &  Y\mbox{ is a continuous u.v.}\\
\log(|\aleph|), &  Y\mbox{ is a discrete u.v.}
\end{cases}
\end{align}
We seek an optimal hypothesis test using the optimization problem in
\begin{align} \label{eqn:optimaltest}
T^*\in \argmax_{T\in \range{H}^{\range{Y}}}   \mathcal{P}(T).
\end{align}

\begin{definition}[Consistency] A test $T:\range{Y}\rightarrow\range{H}$ is consistent if (\textit{i}) $T(Y)=p_0$ only if $Y\in \range{Y|p_0}$ and (\textit{ii}) $T(Y)=p_1$ only if $Y\in \range{Y|p_1}$.
\end{definition}

In the following theorem, we prove that consistent tests are in fact optimal in the sense of $\mathcal{P}$. For any mapping $g:x\mapsto y$, we define the inverse image $g^{-1}(y):=\{x:g(x)=y\}$.

\begin{theorem} \textit{(Optimal Tests):} \label{tho:optimaltest1} Any consistent test is a solution of~\eqref{eqn:optimaltest}.  
\end{theorem}

\begin{proof} First, we proved three important claims. 
\par \textit{Claim~1}: $\range{H|\range{X|y}}=\{p_0,p_1\}$ for any $y\in\range{Y|p_1}\cap \range{Y|p_0}$.
\par The proof for this claim is as follows. For any $y\in\range{Y|p_i}$, there exists $x\in g_Y^{-1}(y)=\range{X|y}$ such that $g_H(x)=p_i$.  Therefore, $\{p_i\}\subseteq g_H(\range{X|y})=\range{H|\range{X|y}}$. This implies that, for any $y\in\range{Y|p_1}\cap \range{Y|p_0}$, $\{p_0,p_1\}\subseteq \range{H|\range{X|y}}\subseteq \range{H}=\{p_0,p_1\}$. 
\par \textit{Claim~2}: $\range{H|\range{X|y}}=\{p_0\}$ for any $y\in\range{Y|p_0}\setminus\range{Y|p_1}$.
\par The proof for this claim is as follows. If $y\notin\range{Y|p_1}$, there must not exist $x\in \range{X|y}$ such that $g_H(x)=p_1$. Therefore, $p_1\notin\range{H|\range{X|y}}$. Therefore, for any $y\in\range{Y|p_0}\setminus\range{Y|p_0}$, it must be that $p_0\in\range{H|\range{X|y}}$ and $p_1\notin\range{H|\range{X|y}}$. This is only possible if $\range{H|\range{X|y}}=\{p_0\}$. 
\par \textit{Claim~3}: $\range{H|\range{X|y}}=\{p_1\}$ for any $y\in\range{Y|p_1}\setminus\range{Y|p_0}$.
\par The proof for this claim is the same as \textit{Claim~2}. 

With these claims in hand, we are ready to prove the lemma. Note that
\begin{align*}
\aleph=&\{y\in\range{Y}:\range{H|\range{X|y}}=\{T(y)\}\}\\
=&\{y\in\range{Y|p_0}\setminus\range{Y|p_1}:\range{H|\range{X|y}}=\{T(y)\}\}\\
&\cup \{y\in\range{Y|p_0}\cap \range{Y|p_1}:\range{H|\range{X|y}}=\{T(y)\}\}\\
&\cup \{y\in\range{Y|p_1}\setminus\range{Y|p_0}:\range{H|\range{X|y}}=\{T(y)\}\}\\
=&\{y\in\range{Y|p_0}\setminus\range{Y|p_1}:\range{H|\range{X|y}}=\{T(y)\}\}\\
&\cup \{y\in\range{Y|p_1}\setminus\range{Y|p_0}:\range{H|\range{X|y}}=\{T(y)\}\},
\end{align*}
where the last equality follows from that, by \textit{Claim~1}, $\range{H|\range{X|y}}=\{p_0,p_1\}$ while $\{T(y)\}$  can be either $\{p_0\}$ or $\{p_1\}$. This shows that
\begin{align*}
\aleph\subseteq (\range{Y|p_0}\setminus\range{Y|p_1})\cup (\range{Y|p_1}\setminus\range{Y|p_0}),
\end{align*}
and as a result
\begin{align*}
\mathcal{P}(T)\leq
\begin{cases}
\log(\mu(\range{Y|p_0}\Delta\range{Y|p_1})), & Y\mbox{ is a continuous u.v.},\\
\log(|\range{Y|p_0}\Delta\range{Y|p_1}|), &Y\mbox{ is a discrete u.v.},
\end{cases}
\end{align*}
where $\Delta$ denotes the symmetric difference operator on the sets. By the definition of consistent tests, we can see that
\begin{align*}
\aleph
=&\{y\in\range{Y|p_0}\setminus\range{Y|p_1}:\range{H|\range{X|y}}=\{p_0\}\}\\
&\cup \{y\in\range{Y|p_1}\setminus\range{Y|p_0}:\range{H|\range{X|y}}=\{p_1\}\},
\end{align*}
Finally,
\begin{align*}
\aleph&=(\range{Y|p_0}\setminus\range{Y|p_1})\cup (\range{Y|p_1}\setminus\range{Y|p_0}),
\end{align*}
because of \textit{Claims~2--3}. This shows that consistent tests attain the upper bound on the performance. 
\end{proof}

Note that if the realization of the lossy/uncertain measurement $Y$ belongs to $\range{Y|p_0}\cap\range{Y|p_1}$, there is not enough evidence to accept or reject the null hypothesis or the alternative hypothesis. However, if the realization of the measurement  $Y$ belongs to $(\range{Y|p_0}\setminus\range{Y|p_1})\cup(\range{Y|p_1}\setminus\range{Y|h_2})=\range{Y|p_0}\Delta\range{Y|p_1}$, with $\Delta$ denoting the symmetric difference operator on the sets, we can confidently reject or accept the null hypothesis or the alternative hypothesis. This fact is used by the consistent tests to achieve the highest performance.

\begin{theorem}\textit{(Performance Bound):}\label{tho:bound}  The performance of any test $T\in \range{H}^{\range{Y}}$ is upper bounded as
\begin{align*}
\mathcal{P}(T)\leq
\begin{cases}
\log(\mu(\range{Y|p_0}\Delta\range{Y|p_1})),& \hspace{-.04in}Y\mbox{ is a continuous u.v.}\\
\log(|\range{Y|p_0}\Delta\range{Y|p_1}|), & \hspace{-0.04in}Y\mbox{ is a discrete u.v.}
\end{cases}
\end{align*}
\end{theorem}

\begin{proof} The upper bound in the statement of the theorem follows from the proof of Theorem~\ref{tho:optimaltest1}. 
\end{proof}

Theorem~\ref{tho:bound} can be seen as a non-stochastic equivalent of Chernoff-Stein Lemma (see, e.g.,~\cite[Ch.\,11]{cover2012elements} for randomized hypothesis testing). Note that $\log(\mu(\range{X|p_0}\Delta\range{X|p_1}))$ essentially captures the difference between the ranges $\range{X|p_0}$ and $\range{X|p_1}$ resembling the Kullback--Leibler divergence in a non-stochastic framework. A same interpretation can also be provided for $\log(|\range{X|p_0}\Delta\range{X|p_1}|)$. 

\begin{example}\textit{(Hypothesis testing using noisy measurements): } Consider an uncertain variable $X=(X_1,X_2)\in[100,250]\times[-10,10]$, where $X_1$ denotes the height of an individual in centimetres and $X_2$ denotes a measurement error in centimetres. Let the uncertain measurement to be $Y=g_Y(X)=X_1+X_2$. Further, the hypothesis uncertain variable is defined as $H=g_H(X)=p_0\mathds{1}_{X_1\leq 150}+p_1\mathds{1}_{X_1>150}$. The null hypothesis $p_0$ is that the individual's is short (shorter than or equal to 150 centimetres) and the alternative hypothesis is that the individual is tall (taller than 150 centimetres). Now, note that
\begin{align*}
\range{Y|p_0}
&=\{X_1+X_2:100\leq X_1\leq 150,X_2\in[-10,10]\}\\
&=[90,160],\\
\range{Y|p_1}
&=\{X_1+X_2:150\leq X_1\leq 250,X_2\in[-10,10]\}\\
&=[140,260].
\end{align*}
Thus $\range{Y|p_0}\cap \range{Y|p_1}=[140,160].$ Let $T$ be a test such that $T(Y)=p_0$ if $Y\in [90,150]$ and $T(Y)=p_1$ if $Y\in(150,260]$.  Evidently, $T$ is a consistent test. We get
\begin{align*}
\mathcal{P}(T)
=&\log(\mu(\range{Y|p_0}\Delta \range{Y|p_1}))\\
=&\log(\mu([90,140)\cup(160,260]))\\
=&\log(150).
\end{align*}
Further, note that $h_0(Y)=\log(170)$, if we scale the performance by $h_0(Y)$, we get
\begin{align*}
\mathcal{P}(T)-h_0(Y)
=&\log(150)-\log(170)\approx -0.1251.
\end{align*}
Now imagine the case where $X=(X_1,X_2)\in[100,250]\times[-20,20]$ with the interpretation that the amount of the additive uncertain measurement noise is twice larger. In this case, we have
\begin{align*}
\mathcal{P}(T)-h_0(Y)=\log(150)-\log(190)\approx-0.2364.
\end{align*}
This shows that, by increasing the amount of the noise, the confidence of the test is reduced, which is in line with our expectation. 
\end{example}

\begin{figure}[t]
\hspace{-.1in}
\begin{tikzpicture}
\node[double=black,draw,rectangle,minimum height=.8cm] (R)  at (+1.5,+0.0) {Adversary};
\node[double=black,draw,rectangle,minimum height=.8cm] (S)  at (-1.0,+0.0) {Sender};
\node[] at (-4.0,+0.0) { \small
\begin{minipage}{2cm}
\centering
hidden \\ uncertainity
\end{minipage}
};
\path[every node/.style={font=\sffamily\small}]
			(S)  		 edge [->,double=black] node[above] {$Y$} 		(R)
			(R)  		 edge [->,double=black] node[above] {$T(X)$}    (+3.5,+0.0)
  (-3,+0.0)  		 edge [->,double=black] node[above] {$H,X$} 	    (S);
\end{tikzpicture}
\caption{\label{fig:diagram0} Communication structure between a sender and a hypothesis-testing adversary.}
\end{figure}
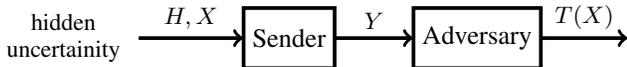

\section{Non-Stochastic Privacy Against Hypothesis-Testing Adversary}
\label{sec:privacy}
Consider the communication diagram in Figure~\ref{fig:diagram0} between a sender  and an adversary. The adversary's ultimate aim is to accurately test a hypothesis $H$ based on the communicated information from the sender $Y$. The sender wants to provide a message $Y$ that is as close as possible to $X$ while making the adversary's task in testing the validity of hypothesis $H$ hard. The policy of the sender is captured by the mapping from $X$ to $Y$, denoted by $g_Y$. We use the performance of the adversary in testing the private hypothesis based on the reported output $Y$ to define a measure of privacy as
\begin{align}
\mathrm{Priv}(g_Y):=
\begin{cases}
h_0(Y)-\log(\mu(\range{Y|p_0}\Delta\range{Y|p_1})), &\\
& \hspace{-1in}Y\mbox{ is a continuous u.v.}\\
h_0(Y)-\log(|\range{Y|p_0}\Delta\range{Y|p_1})), \\
& \hspace{-1in}Y\mbox{ is a discrete u.v.}
\end{cases}
\end{align}
Note that increasing $\mathrm{Priv}(g_Y)$ implies that the size of $\range{Y|p_0}\Delta\range{Y|p_1}$ is decreased, thus degrading the performance of any test employed by the adversary in light of Theorem~\ref{tho:bound}. 

\begin{definition}[$\epsilon$-privacy] Policy $g_Y$ is $\epsilon$-private for some $\epsilon\in(0,1]$ if  $\mathrm{Priv}(g_Y)\geq \log(\epsilon)$. 
\end{definition}

We need to balance privacy with utility, otherwise the best policy is to report nothing. Therefore, we need to define a measure of accuracy to balance against the privacy.

\begin{definition}[$\rho$-accuracy] Policy $g_Y$ is $\rho$-accurate for some $\rho\in(0,+\infty)$ if $\sup_{X\in\range{X}}\|X-g_Y(X)\|\leq 1/\rho$. 
\end{definition}

Increasing $\rho$ in $\rho$-accuracy implies that $\sup_{X\in\range{X}}\|X-Y\|$ is decreased, thus improving the quality of the reported output $Y$ by enforcing it to stay consistently closer to $X$. 

\begin{theorem} \label{tho:privacy_versus_utility} Assume that $\range{X}\subseteq\mathbb{R}^{n_x}$, and $g:\mathbb{R}^{n_x-1}\rightarrow\mathbb{R}$ exists such that 
\begin{align*}
g_H(x)
=
\begin{cases}
p_0, & x_i-g(x_{-i})\geq 0,\\
p_1, & x_i-g(x_{-i})< 0,
\end{cases}
\end{align*}
where $x_{-i}=(x_j)_{j\neq i}$. Let
\begin{align*}
g_Y(x)
=
\begin{cases}
(g(x_{-i}),x_{-i}), & \displaystyle  g(x_{-i})-\frac{1}{\rho}\leq x_i\leq g(x_{-i})+\frac{1}{\rho},\\
x, & \mbox{otherwise},
\end{cases}
\end{align*}
and
\begin{align*}
\epsilon
=\frac{\displaystyle\mu\left(\range{X}\cap \left\{x:g(x_{-i})-\frac{1}{\rho}\leq x_i\leq g(x_{-i})+\frac{1}{\rho}\right\}\right)}{\exp(h_0(X))}.
\end{align*}
Then, $g_Y$ is $\rho$-accurate and $\epsilon$-private.
\end{theorem}

\begin{proof} The proof for $\rho$-accuracy, with $\rho\in(0,+\infty)$, follows from that $\|X-g_Y(X)\|=|x_i-g(x_{-i})|\leq 1/\rho$. The proof for $\epsilon$-privacy follows from that, if $\range{Y}$, $(\range{Y|p_0}\Delta\range{Y|p_1})$, $(\range{Y|p_0}\cap\range{Y|p_1})$ are Lebesgue measurable, we get $\mu(\range{Y})
=\mu((\range{Y|p_0}\Delta\range{Y|p_1})\cup (\range{Y|p_0}\cap\range{Y|p_1}))=\mu(\range{Y|p_0}\Delta\range{Y|p_1})+\mu(\range{Y|p_0}\cap\range{Y|p_1})$
because $(\range{Y|p_0}\Delta\range{Y|p_1})\cup (\range{Y|p_0}\cap\range{Y|p_1})
=\range{Y|p_0}\cup\range{Y|p_1}
=\range{Y}.$
\end{proof}

For large enough $\rho$, it can be seen that
$\range{X}\cap \{x:g(x_{-i})-\rho^{-1}\leq x_i\leq g(x_{-i})+\rho^{-1}\}\approx\{x:g(x_{-i})-\rho^{-1}\leq x_i\leq g(x_{-i})+\rho^{-1}\},$ and, as a result, $\epsilon=\mathcal{O}(\rho^{-1})$. This implies that, for the policy in Theorem~\ref{tho:privacy_versus_utility}, we have
\begin{center}
\mybox{$\mbox{``privacy}\times \mbox{accuracy} = \mbox{constant''}.$}
\end{center}
With the theoretical results in hand, we can demonstrate the effects of using the policy in Theorem~\ref{tho:privacy_versus_utility} on a practical dataset in the next section.

\section{Numerical Results}
\label{sec:numerical}
In this subsection, we consider design of a privacy-preserving policy for reporting individuals height in centimetres and weight in kilograms publicly. We consider an adversary who is interested in identifying individuals passing the obesity threshold in terms of body mass index (BMI), e.g., an insurance agency may use publicly available data to increase premiums of obese people or deny them insurance. Therefore, there is a duty of care for releasing demographic data of individuals publicly.  By the definition of the U.S. Department of Health \& Human Services, a person, be it female or male, is considered obese if their BMI is greater than or equal to 30. 

\begin{figure}
\includegraphics[width=1\linewidth]{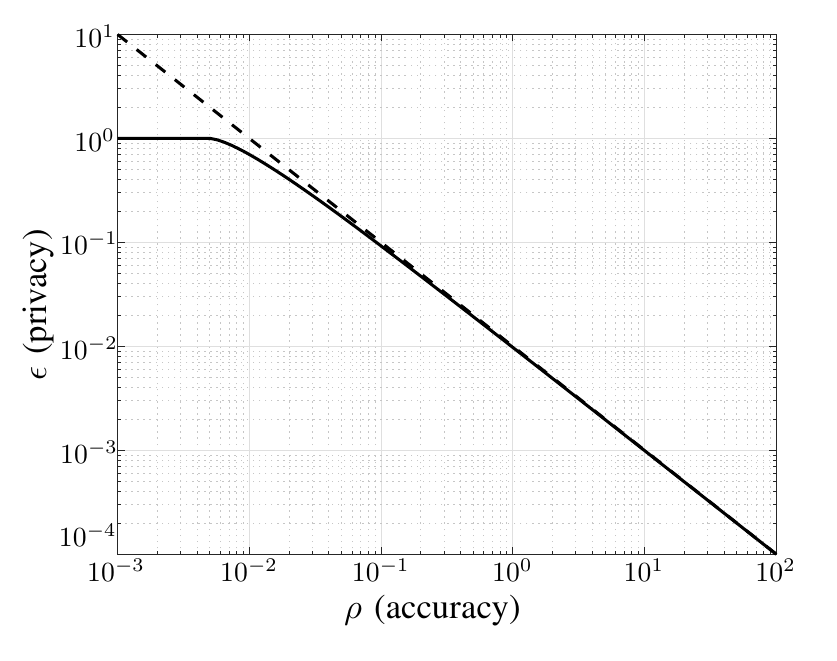}
\caption{
\label{fig:1}
Privacy guarantee, $\epsilon$, versus accuracy level, $\rho$. The dashed line shows the asymptotic $\mathcal{O}(\rho^{-1})$.
}
\end{figure}

\begin{figure*}
\centering
\includegraphics[width=1\linewidth]{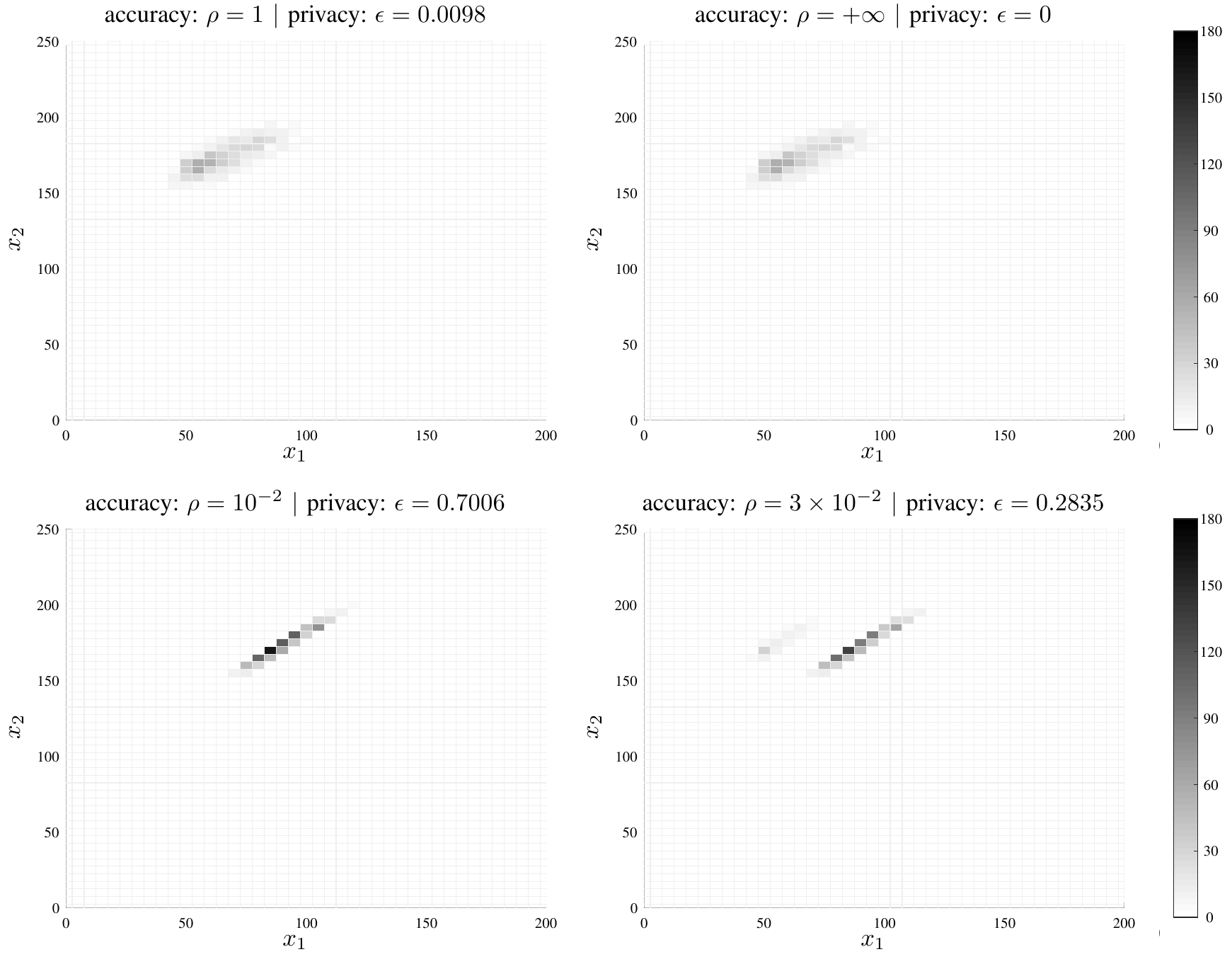}
\caption{
\label{fig:2}  
The histogram of the reported weight and hight of individuals for various levels of accuracy, $\rho$. The darker colors show higher frequencies.}
\end{figure*}

Let uncertain variable $X$ denote the weight and height, i.e., $X=[X_1\;X_2]^\top$ with $X_1\in[0,200]$ denoting the weight in kilograms and $X_2\in[0,250]$ denoting the height in centimetres. The hypothesis is given by
\begin{align*}
g_H(x)=
\begin{cases}
p_0, & \displaystyle \frac{x_1}{(x_2/100)^2}\geq  30,\\[1em]
p_1, & \displaystyle \frac{x_1}{(x_2/100)^2}< 30.
\end{cases}
\end{align*}
Following the notation of Theorem~\ref{tho:privacy_versus_utility}, we can redefine the hypothesis using the sign of $x_1-g(x_2)$ with $g:x_2\mapsto 30(x_2/100)^2$. Define
\begin{align} \label{eqn:bmiprivacypolicy}
g_Y(x)
=
\begin{cases}
30(x_2/100)^2, & \displaystyle  \frac{3x_2^2}{1000}-\frac{1}{\rho}\leq x_1 \leq \frac{3x_2^2}{1000}+\frac{1}{\rho},\\
x_1, & \mbox{otherwise}.
\end{cases}
\end{align}
Note that
\begin{align*}
&\mu\left(\range{X}\cap \left\{x:g(x_{-i})-\frac{1}{\rho}\leq x_i\leq g(x_{-i})+\frac{1}{\rho}\right\}\right)\\
&\hspace{.8in}=\int_{0}^{250} \int_{\max( 3x_2^2/1000-1/\rho,0)}^{\min(3x_2^2/1000+1/\rho,200)}
\mathrm{d}x_1\mathrm{d}x_2\\
&\hspace{.8in}=\int_{0}^{250} [\max( 3x_2^2/1000-1/\rho,0)\\
&\hspace{1.3in}-\min(3x_2^2/1000+1/\rho,200)]\mathrm{d}x_2
\end{align*}
Using this, we can compute the level of privacy guarantee. The solid black line in Figure~\ref{fig:1} illustrates privacy guarantee, $\epsilon$, versus accuracy level, $\rho$. The dashed line shows the asymptotic $\mathcal{O}(\rho^{-1})$. As expected, by increasing accuracy, the privacy guarantee can only be reduced and \textit{vice versa}.

\begin{figure}
\centering
\includegraphics[width=1\linewidth]{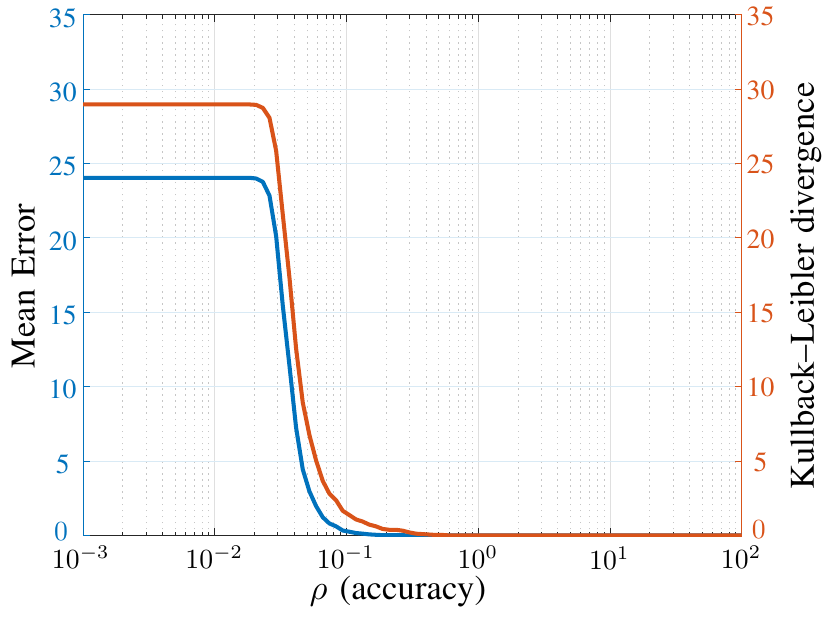}
\caption{
\label{fig:3}
The blue curve illustrates the difference of the mean value of the weight of the population with and without the privacy preserving policy. The red curve shows the difference between the empirical density functions of the weight across the population with and without the privacy preserving policy, measured by the Kullback–Leibler divergence.
}
\end{figure}

Now, we use a real dataset to investigate the effects of the privacy-preserving policy in~\eqref{eqn:bmiprivacypolicy}. We use a dataset of preferences, interests, and demographics of young people, aged between 15-30, of Slovakian nationality~\cite{kaggle_slovak}. The data was gathered in 2013 by students of an statistics class at FSEV UK through friends and families. The dataset consists of 1,010 records with 150 features (139 integer and 11 categorical) including height, weight, music preferences, eating habits, etc. This dataset is  popular for analysis on Kaggle (an online platform for sharing data) with 99.4k views and 21.8k downloads on all continents within the last three years. Noting that the preferences of the individuals can be matched with publicly available datasets, such as IMDb (Internet Movie Database), to potentially identify the individuals, there is a need for obfuscating the data in order to avoid privacy breaches related to age, weight, and height. 

Assume that we use the policy~\eqref{eqn:bmiprivacypolicy} is used for reporting weight and height of individuals so that potential future insurance agencies cannot test for the obesity levels. Figure~\ref{fig:2} illustrates the histogram of the reported weight and hight of individuals for various levels of accuracy, $\rho$. The darker colors show higher frequencies. As expected, by decreasing $\rho$, the accuracy gets worse; the histogram changes more drastically. The blue curve in Figure~\ref{fig:3} illustrates the difference of the mean value of the weight of the population with and without the privacy preserving policy. The red curve in Figure~\ref{fig:3} shows the difference between the empirical density functions of the weight across the population with and without the privacy preserving policy. These curves allows a data curator to balance between privacy and utility.

\section{Conclusions and Future Works}
\label{sec:conclusions}
We considered privacy against hypothesis testing adversaries using the theory of non-stochastic hypothesis testing. We constructed reporting policies with prescribed privacy and utility guarantees. We demonstrated the utility of privacy-preserving policies on a real dataset. Future work can focus on development of optimal policies.

\bibliographystyle{ieeetr}
\bibliography{scibib}

\end{document}